\documentclass[11pt,reqno]{amsproc}
\usepackage{amsfonts}
\usepackage{amsthm}
\usepackage{amsmath}
\usepackage{amscd}
\usepackage{amsbsy}
\usepackage{epsfig,afterpage}
\usepackage{graphicx}
\usepackage{color}

\newtheorem {theorem} {Theorem}
\newtheorem {definition} [theorem]{Definition}
\newtheorem {proposition} [theorem]{Proposition}

\newtheorem {lemma}  [theorem]{Lemma}
\newtheorem {remark} [theorem]{Remark}

\input xy
\xyoption{all}

\begin{document}
\title[Liouvillian Propagators, Riccati Eq. \& Diff. Galois Theory]{%
Liouvillian Propagators, Riccati Equation and Differential Galois Theory}
\author{Primitivo Acosta-Hum\'anez}
\address[P. Acosta-Hum\'anez]{Departamento de Matem\'aticas y Estad\'istica
Universidad del Norte Km. 5 via Puerto Colombia, Barranquilla, Colombia}
\email[P. Acosta-Hum\'anez]{pacostahumanez@uninorte.edu.co}
\author{Erwin Suazo}
\address[E. Suazo]{University of Puerto Rico at Mayag\"{u}ez Department of
Mathematical Sciences CALL BOX 9000 Mayag\"{u}ez, Puerto Rico 00681-9018}
\email[E. Suazo]{erwin.suazo@upr.edu}
\dedicatory{To Victoria and Joaqu\'{\i}n, propagators of our happiness}

\begin{abstract}
In this paper a Galoisian approach to build propagators through Riccati
equations is presented. The main result corresponds to the relationship
between the Galois integrability of the linear Schr\"odinger equation and the
virtual solvability of the differential Galois group of its associated
characteristic equation. As main application of this approach we solve the Ince's differential equation through Hamiltonian Algebrization procedure and Kovacic Algorithm to find the propagator for a
generalized harmonic oscillator that has applications describing the process
of degenerate parametric amplification in quantum optics and the description of 
the light propagation in a nonlinear anisotropic waveguide. Toy models of propagators inspired by integrable Riccati equations and integrable characteristic equations are also presented. 
\end{abstract}

\maketitle

\section{Introduction}

The generalized harmonic oscillator has attracted considerable attention
over many years in view of their great importance to several advanced
quantum problems including Berry's phase, quantization of mechanical of
systems and more (see \cite{cosuin} and references therein). The fact that
in quantum electrodynamics the electromagnetic field can be represented as a
set of forced harmonic oscillators makes quadratic Hamiltonians of special
interest \cite{cosusu1, cosusu2, DoMaMa, Hannay, su1, YeLeePan}. A method to
construct explicit propagators for the linear Schr\"{o}dinger equation with
a time-dependent quadratic Hamiltonian based in solutions of the Riccati
equation has been presented in \cite{colosusu, cosusu1}, and as particular
cases the propagators for the free particle, the harmonic oscillator,
Caldirola-Kanai and more can be found in a unified manner \cite{colosusu,
cosusu1, LaSu}. Cauchy problem with this approach have been studied in \cite%
{LoSu, MeCoSu}.

There has also been an increase interest in the study of Picard-Vessiot
Theory, also known as the Galois theory for ordinary linear differential
equations, where differential equations of such type are analysed throughout
their Galoisian structure. This Galoisian structure depends on the nature of
the solutions of the differential equation; for instance, one obtains some
kind of solvability (virtual) for the \emph{Galois group} whenever one
obtains \emph{Liouvillian solutions}, and in this case one says that the
differential equation is \emph{integrable}. This means for example that when
one obtains Airy functions, the differential equation is not integrable,
while when one obtains Jacobi elliptic functions, the differential equation
is integrable, and for instance one gets virtual solvability of its Galois
group.

In this paper we present  a Galoisian approach of  how to find explicit
propagators through Liouvillian solutions for linear second order
differential equations associated to Riccati equations. The main application of this Galoisian approach,  for instance a main result of
the paper, is the construction of the propagator for the so called degenerate
parametric oscillator:

\begin{eqnarray}
i\partial _{t}\psi  &=&H(t)\psi   \label{SEsu1} \\
H(t) &=&\frac{1}{2m}\left( 1+\frac{\lambda }{\omega }\cos (2\omega t)\right)
p^{2}+\frac{m\omega ^{2}}{2}\left( 1-\frac{\lambda }{\omega }\cos (2\omega
t)\right) x^{2}  \label{SEsu2} \\
&&+\frac{\lambda }{2}\sin (2\omega t)(px+xp),\quad p=-i\partial _{x}.  \notag
\end{eqnarray}

Another form for $H(t)$ in terms of annihilation and creation operators, $%
\hat{a}=\sqrt{1/2\omega }\left( \omega x+ip\right) ,\quad \hat{a}^{+}=\sqrt{%
1/2\omega }(\omega x-ip)$ with $\hat{a}\hat{a}^{+}-\hat{a}^{+}\hat{a}=1,$ in
the coordinate representation is:%
\begin{equation}
H(t)=\frac{\omega }{2}\left( \hat{a}\hat{a}^{+}+\hat{a}^{+}\hat{a}\right) -%
\frac{\lambda }{2}\left( e^{2i\omega t}\hat{a}^{2}-e^{-2i\omega t}(\hat{a}%
^{+})^{2}\right) .
\end{equation}

In quantum optics the first term corresponds to the self-energy of the
oscillator representing the mode of interest, and the second term describes
the coupling of the classical pump to that mode giving rise to the
parametric amplification process ($\lambda $ is the phenomenological
constant) \cite{cosuin}. The oscillator (\ref{SEsu1})-(\ref{SEsu2}) might be
introduced for the first time by Takahasi \cite{Taka} in order to describe
the process of degenerate parametric amplification in quantum optics (see
also \cite{LiFaXi, Lo1, Loui2, Loui3, Orzag, Shen, Taka}). The Hamiltonian (%
\ref{SEsu2}) had also been considered by Angelow and Trifonov \cite{AnTri95,
AnTri98} in order to describe the light propagation in a nonlinear
anisotropic waveguide. 

In  \cite{cosuin} the authors motivated the investigation of properties of
the degenerate parametric oscillator (a particular case of the generalized
harmonic oscillator) including a systematic study of corresponding
non-periodics solutions of Ince's equation that seems to be missing in the
mathematical literature (see section 8 in  \cite{cosuin}). In \cite{cosuin}
the authors constructed the propagator for (\ref{SEsu1})-(\ref{SEsu2}) for
the case $\lambda =\omega =1$ using non-periodic solutions of Ince's
equation (compare with classical results of periodic solutions \cite{MaWi});

In this paper using Galoisian theory we show how to find the explicit
solution of the associated Ince's equation presented in \cite{cosuin} for
the special case $\lambda =\omega $ (those authors don't show how they found
the explicit solution, see equations (6.1), (6.2) and (6.3)). In fact, they
present a non-periodic solution that allows us to write the propagator
explicitly; the fact of the solution being non-periodic is fundamental. In
this note we present in detail how to find this non-periodic solution by
using a combination of kovacic algorithm \cite{kov} and an algebrization
procedure (see \cite{ac,acmowe})$,$ and further we find the explicit
solution for the general case  $\lambda \neq \omega $ (see (3.4) in \cite%
{cosuin} and discussion below of it)$.$ We believe this approach can be
extended to the study of propagators of other generalized harmonic
oscillators, but here we restrict ourselves to (\ref{SEsu1})-(\ref{SEsu2})
and give some toy examples in Section \ref{toyexamples}.

The aim of this paper is to establish a Galoisian approach to the techniques
given by Suslov et.al., see \cite{colosusu,cosusu1, cosusu2,cosu}. To study
Liouvillian solutions for linear second order differential equations, as
well the integrability of their associated Riccati equations, we use Kovacic
algorithm (see \cite{kov}) and an algebrization procedure (see \cite%
{ac,acmowe}). These tools were applied to study differential equations
incoming from physics; in particular the integrability analysis of the
one-dimensional linear Schr\"{o}dinger equation has been studied in \cite%
{ac,acmowe}.

This paper is organized in the following way:

\begin{itemize}
\item Section 2 contains a brief description of the basic theory concerning
of construction of explicit propagators using the references \cite{colosusu,
cosusu1, cosusu2, cosu, cosuin, susu1, susu2, su1, su2}, and a short summary
of the Picard-Vessiot theory is also presented, which was written according
to \cite{ac,aclamopa,vasi}.

\item Section \ref{Galois_prop} contains the one of the main results of this paper; it corresponds to a Galoisian approach of propagators. It is devoted to a theoretical Galoisian approach to propagators starting
with Riccati and second order differential equations. The result given here
relates the Galois integrability of the linear Schr\"{o}dinger equation with
the virtual solvability of the differential Galois group of its associated
characteristic equation.

\item Section \ref{newpropagators} contains the another main result of this paper, where  is presented the Galoisian
analysis of Ince's differential equations to do the construction of the propagator of the degenerate parametric oscillator in more general terms.   Also, we compute the differential Galois
group associated to such propagator, which corresponds to the differential
Galois group of the Ince's characteristic equation.

\item Section \ref{toyexamples} contains some toy models of 
new propagators, as well their Green functions through characteristic
equations and Riccati equations. 
\end{itemize}

For suitability, along this paper $\partial _{x}$ denotes $\frac{\partial }{%
\partial x}$, for higher order derivation $\partial _{x}^{n}$ denotes $\frac{%
\partial ^{n}}{\partial x^{n}}$ and by $\partial_ta(0)$ we mean $%
\partial_ta(t)|_{t=0}$.

\section{Theoretical Background}\label{thback}

\subsection{Differential Galois theory}

\label{DGT} The Galois theory of differential
equations, also called Differential Galois Theory and Picard-Vessiot Theory, has been developed by
Picard, Vessiot, Kolchin and currently by a lot of researchers, see \cite{aclamopa,acmowe,ka,kol,kov,mara,vasi}. Moreover, recent applications to mathematical physics can be found in \cite{ac,acmowe,mo,mora,stpr}. We consider the differential Galois theory in the context of second order linear differential equations.

A \emph{differential field} $K$ is a commutative field of
characteristic zero with a \emph{derivation}  $\partial
_{x}$, where the \emph{field of constants} of $K$, denoted by $\mathcal{C}_K$, is  algebraically
closed and of characteristic zero.
The \emph{coefficient field} for a differential equation is defined
as the smallest differential field containing all the coefficients of the
equation.
Let $L$ be a differential field
containing $K$, we say that $L$ is a \emph{%
Picard-Vessiot} extension of $K$ if there exist two
linearly independent $y_{1},y_{2}\in L$ such that
$L=K\langle y_{1},y_{2}\rangle $ and $\mathcal C_L=\mathcal C_K$.
A $K$-automorphism $\sigma $ of the
Picard-Vessiot extension $L$ is called a \emph{differential automorphism} if $\sigma
(\partial _{x}a)=\partial _{x}(\sigma (a))$ for all $a\in L$ and $\forall
a\in K,$ $\sigma (a)=a$. The group of all differential automorphisms of $L$
over $K$ is called the \textit{differential Galois group} of $L$ over $K$
and is denoted by $\mathrm{DGal}(L/K)$. 
\begin{theorem}\cite{kol}
The differential Galois group $\mathrm{DGal}(L/K)$ is an algebraic
subgroup of $\mathrm{GL}(2,\mathbb{C})$.
\end{theorem}

We denote by $G^0$ the \emph{connected component of the identity}, thus, when $G^0$ satisfies some property, we say that $G$ virtually
satisfies such property.

\begin{theorem}[Lie-Kolchin Theorem]\label{liekolchin} Let $G\subseteq \mathrm{GL}(2,\mathbb{C}%
) $ be a virtually solvable group. Then $G^0$ is triangularizable.
\end{theorem}

We say that a linear differential equation $%
\mathcal{L}$ is  \textit{integrable} if the Picard-Vessiot
extension $L\supset K$ is obtained as a tower of differential fields $%
K=L_0\subset L_1\subset\cdots\subset L_m=L$ such that $L_i=L_{i-1}(\eta)$
for $i=1,\ldots,m$, where either

\begin{enumerate}
\item $\eta$ is {\emph{algebraic}} over $L_{i-1}$, that is $\eta$ satisfies
a polynomial equation with coefficients in $L_{i-1}$.

\item $\eta$ is {\emph{primitive}} over $L_{i-1}$, that is $\partial_x\eta
\in L_{i-1}$.

\item $\eta$ is {\emph{exponential}} over $L_{i-1}$, that is $\partial_x\eta
/\eta \in L_{i-1}$.
\end{enumerate}

The solutions obtained throgh such towers are called 
\textit{Liouvillian}. The \emph{special functions} are not always Liouvillian, we can see that Airy equation has not Liouvillian solutions, while Bessel equation has Liouvillian solutions for special values of the parameter, see \cite{ince,niuv}. Thus, we mean by integrable whenever the differential equation has Liouvillian solutions instead of special functions solutions. 

\begin{theorem}[Kolchin]\label{kolchin}
A linear differential equation  is integrable if and only if $\mathrm{DGal}%
(L/K)$ is virtually solvable.
\end{theorem}

\begin{proposition}[Riccati transformations, \protect\cite{aclamopa}]
\label{prop:almp} Let $K$ be a differential field, 
$a_0(x)$, $a_1(x)$, $a_2(x)$, $r(x)$, $\rho(x)$, $b_0(x)$, $b_1(x)$
belonging to $K$. Consider now the following forms associated to
any second order differential equation (ode) and Riccati equation:

\begin{itemize}
\item[$(i)$] Second order ode (in general form):
\begin{equation}  \label{lemma:2ode:initial_form}
\partial_x^2y+b_1 \partial_xy+ b_0 y=0.
\end{equation}

\item[$(ii)$] Second order ode (in reduced form):
\begin{equation}  \label{lemma:2ode:normal_form}
\partial_x^2\xi=\rho\xi.
\end{equation}

\item[$(iii)$] Riccati equation (in general form):
\begin{equation}  \label{lemma:riccati:initial_form}
\partial_xv=a_0 + a_1 v + a_2 v^2 ,\quad a_2 \neq 0.
\end{equation}

\item[$(iv)$] Riccati equation (in reduced form):
\begin{equation}  \label{lemma:riccati:normal_form}
\partial_xw=r-w^2,
\end{equation}
\end{itemize}

Then, there exist transformations $T$, $B$, $S$ and $R$ leading some of
these equations into the other ones, as showed in the following diagram:
\begin{equation*}
\xymatrix{ \partial_xv=a_0 + a_1 v + a_2 v^2 \ar[r]^-{T} & \partial_xw=r - w^2 \\ \partial_x^2y+b_1 \partial_xy +
b_0 y=0 \ar[r]^-{S} \ar@{<-}[u]^{B} & \ar@{<-}[u]_{R} \partial_x^2\xi=\rho \xi. }
\end{equation*}%
The new independent variables are defined by means of
\begin{equation*}
\begin{array}{lcl}
T:\ v=-\left( \dfrac{\partial_xa_{2}}{2a_{2}^{2}}+\dfrac{a_{1}}{2a_{2}}%
\right) -\dfrac{1}{a_{2}}w, & \qquad \qquad & B:\ {\displaystyle v=-\frac{1}{%
a_{2}}\frac{\partial_xy}{y}}, \\
&  &  \\
\ S:\ {\displaystyle y=\xi \mathrm{e}^{-{\frac{1}{2}}\int b_{1}dx}}, &
\qquad \qquad & R:\ {\displaystyle w=\frac{\partial_x\xi}{\xi }},%
\end{array}%
\end{equation*}%
and the functions $r$, $\rho $, $b_{0}$ and $b_{1}$ are given by
\begin{align}
r& =\frac{1}{\beta }\left( a_{0}+a_{1}\alpha +a_{2}\alpha ^{2}-\partial_x\alpha\right) , \\
\alpha & =-\left( {\frac{\partial_xa_{2}}{2a_{2}^{2}}}+{\frac{a_{1}}{2a_{2}}}%
\right) ,\quad \beta =-\frac{1}{a_{2}}, \\
b_{1}& =-\left( a_{1}+{\frac{\partial_xa_{2}}{a_{2}}}\right) ,\quad
b_{0}=a_{0}a_{2}, \\
\rho & =r=\frac{b_{1}^{2}}{4}+\frac{\partial_xb_{1}}{2}-b_{0}.
\end{align}
\end{proposition}
\begin{remark}\label{rem:almp}
From Proposition \ref{prop:almp}, it is recovered the well known result in differential Galois theory (see for example \cite{vasi}):
Riccati equation has an algebraic solution over the
differential field $K$ if and only if its associated second order differential equation has two independent Liouvillian solutions
(the differential Galois group of the second order differential equation is
virtually solvable). Furthermore, the differential Galois group for equation \eqref{lemma:2ode:normal_form} is a subgroup of $\mathrm{SL}(2,\mathbb C)$, such is the case for the stationary Schr\"odinger equation.
\end{remark}

Kovacic (see \cite{kov} and improvements in \cite{ac,acmowe}) developed an algorithm to solve the
differential equation $$\partial _{\tau}^{2} y =ry, \quad r=\frac{f}{g},\quad f,g\in \mathbb{
C}[\tau].$$ There are four cases in Kovacic's algorithm. Only for cases 1, 2 and 3 we
can solve the differential equation, but for the case 4 the differential
equation is not integrable. We use the following notations: $$\Gamma^{\prime}=\left\{ c\in\mathbb{C}:\,g(c)=0\right\}, \quad \Gamma=\Gamma^{\prime}\cup\{\infty\},$$ by $\circ(r_c)$, we mean
the multiplicity of $c$ as a pole of $r$, while by  $\circ\left( r_{\infty}\right) $ we
mean the order of $\infty$ as a zero of $r$. Now, we summarize the cases one and two of Kovacic Algorithm that will be used in Section \ref{newpropagators}.

\textbf{Case 1.}  

\textit{Step 1.}

\begin{description}
\item[$(c_{2})$] If $\circ\left( r_{c}\right) =2,$ and
\begin{equation*}
r= \cdots + b(\tau-c)^{-2}+\cdots,\quad \text{then}
\end{equation*}
\begin{equation*}
\left[ \sqrt {r}\right]_{c}=0,\quad \alpha_{c}^{\pm}=\frac{1\pm\sqrt{1+4b}}{2%
}.
\end{equation*}

\item[$(\infty_{2})$] If $\circ\left( r_{\infty}\right) =2,$ and $r= \cdots
+ \tau^{2}+\cdots$, then
\begin{equation*}
\left[ \sqrt{r}\right] _{\infty}=0,\quad\alpha_{\infty}^{\pm}=\frac{1\pm%
\sqrt{1+4b}}{2}.
\end{equation*}
\end{description}

\textit{Step 2.} Find $D\neq\emptyset$ defined by
\begin{equation*}
D=\left\{ n\in\mathbb{Z}_{+}:n=\alpha_{\infty}^{\varepsilon (\infty)}-{%
\displaystyle\sum\limits_{c\in\Gamma^{\prime}}} \alpha_{c}^{\varepsilon(c)},%
\forall\left( \varepsilon\left( p\right) \right) _{p\in\Gamma}\right\} .
\end{equation*}
If $D=\emptyset$, then we should start with the case 2. Now, if $\mathrm{Card%
}(D)>0$, then for each $n\in D$ we search $\omega$ $\in\mathbb{C}(\tau)$ such
that
\begin{equation*}
\omega=\varepsilon\left( \infty\right) \left[ \sqrt{r}\right] _{\infty}+{%
\displaystyle\sum\limits_{c\in\Gamma^{\prime}}} \left( \varepsilon\left(
c\right) \left[ \sqrt{r}\right] _{c}+{{\alpha_{c}^{\varepsilon(c)}}\over {(\tau-c)}%
}\right).
\end{equation*}

\textit{Step 3}. For each $n\in D$, search for a monic polynomial $P_n$ of
degree $n$ with
\begin{equation*}  \label{recu1}
\partial_\tau^2P_n + 2\omega \partial_\tau P_n + (\partial_\tau\omega + \omega^2 - r)
P_n = 0.
\end{equation*}
If success is achieved then $y_1=P_n e^{\int\omega}$ is a solution of
the differential equation. Else, case 1 cannot hold. \bigskip

\textbf{Case 2.}

\textit{Step 1.} Search for each $c\in\Gamma^{\prime}$ and $\infty$ the sets
$E_{c}\neq\emptyset$ and $E_{\infty}\neq\emptyset.$ For each $%
c\in\Gamma^{\prime}$ and for $\infty$ we define $E_{c}\subset\mathbb{Z}$ and
$E_{\infty}\subset\mathbb{Z}$ as follows:

\begin{description}
\item[($c_2$)] If $\circ\left( r_{c}\right) =2,$ and $r= \cdots +
b(\tau-c)^{-2}+\cdots ,\ $ then
\begin{equation*}
E_{c}=\left\{ 2+k\sqrt{1+4b}:k=0,\pm2\right\}\cap\mathbb{Z}.
\end{equation*}

\item[$(\infty_{2})$] If $\circ\left( r_{\infty}\right) =2,$ and $r= \cdots
+ b\tau^{2}+\cdots$, then
\begin{equation*}
E_{\infty }=\left\{ 2+k\sqrt{1+4b}:k=0,\pm2\right\}\cap\mathbb{Z}.
\end{equation*}
\end{description}

\textit{Step 2.} Find $D\neq\emptyset$ defined by
\begin{equation*}
D=\left\{ n\in\mathbb{Z}_{+}:\quad n=\frac{1}{2}\left( e_{\infty}- {%
\displaystyle\sum\limits_{c\in\Gamma^{\prime}}} e_{c}\right) ,\forall
e_{p}\in E_{p},\quad p\in\Gamma\right\}.
\end{equation*}
If $D=\emptyset,$ then we should start the case 3. Now, if $\mathrm{Card}%
(D)>0,$ then for each $n\in D$ we search a rational function $\theta$
defined by
\begin{equation*}
\theta=\frac{1}{2} {\displaystyle\sum\limits_{c\in\Gamma^{\prime}}} \frac{%
e_{c}}{\tau-c}.
\end{equation*}

\textit{Step 3.} For each $n\in D,$ search a monic polynomial $P_n$ of
degree $n$, such that {\small {\
\begin{equation*}  \label{recu2}
\partial_\tau^3P_n+3\theta \partial_\tau^2P_n+(3\partial_\tau\theta+3\theta
^{2}-4r)\partial_\tau P_n+\left( \partial_\tau^2\theta+3\theta\partial_\tau\theta
+\theta^{3}-4r\theta-2\partial_\tau r\right)P_n=0.
\end{equation*}%
}} If $P_n$ does not exist, then case 2 cannot hold. If such a polynomial is
found, set $\phi = \theta + \partial_xP_n/P_n$ and let $\omega$ be a
solution of
\begin{equation*}
\omega^2 + \phi \omega + {\frac{1}{2}}\left(\partial_\tau\phi + \phi^2
-2r\right)= 0.
\end{equation*}

Then $y_1 = e^{\int\omega}$ is a solution of the differential equation.
\bigskip

We can see that applying by hand Kovacic Algorithm it can be a little difficults, but thanks
to has been implemented it in Maple (command \texttt{kovacicsols}) we can avoid such calculations. The
problem with Maple is that the answers can have complicated expressions that
should be transformed into more suitable and readable expressions. On the
other hand, Kovacic Algorithm only works with rational coefficients, for
instance, when the differential equation has not rational coefficients we cannot apply Kovacic Algorithm.

 We recall that in most of the cases is better the application of Kovacic Algorithm by
hand (when can appear parameters or the solution is very complicated), for
example the stationary Schr\"{o}dinger equation $\partial _{x}^{2}\Psi
=(x^{2}-\lambda )\Psi $ cannot be solved with \texttt{kovacicsols} due to
the parameter $\lambda$, thus, the output given by Maple is $[,]$, i.e., we
should solve the equation using Kovacic Algorithm by hand to obtain solutions conditioned to some values of $\lambda$. The interested
reader can see in \cite{ac,acmowe} the solutions of Schr\"{o}dinger
equations (stationary) using Kovacic Algorithm by hand, recalling that was
not possible the obtaining of such solutions through \texttt{kovacicsols}.

One question is concerning to the application of Kovacic Algorithm whether the coefficients are not rational functions, this problem can be solved using \emph{Hamiltonian Algebrization}, procedure developed in \cite{ac,acmowe}. We present a short summary of the Hamiltonian algebrization process that will be used in Section \ref{newpropagators}.
Following \cite{ac,acmowe}, we say that $\tau=\tau(t)$ is a \emph{Hamiltonian change of variable} whether $(\tau,\partial_t\tau)$ is a solution curve of the Hamiltonian $$H=\frac{p^2}{2}+V(\tau), \quad \partial_t\tau=\partial_pH=p,\quad \partial_tp=-\partial_\tau H=-\partial_\tau V(\tau),\quad V(\tau)\in\mathbb{C}(\tau).$$  Thus, we can denote by $\alpha$ to $p^2$, which is dependent of $\tau$, that is $$\alpha=2H-2V(\tau)=(\partial_t\tau)^2,\quad \partial_t\tau=\sqrt{\alpha}.$$ In this way, as particular case because the theory is more general, we can transform differential equations $$\partial_t^2\mu+p\partial_t\mu+q\mu=0\leadsto \widehat\partial_\tau^2\widehat\mu+\widehat p\widehat\partial_\tau\widehat\mu+\widehat q\widehat\mu=0,$$ where $\widehat\partial_\tau=\sqrt{\alpha}\partial_\tau$, $\widehat \mu\circ \tau=\mu$, $\widehat p\circ \tau=p$, $\widehat q\circ \tau=q$. Moreover, the differential equation  $\widehat\partial_\tau^2\widehat\mu+\widehat p\widehat\partial_\tau\widehat\mu+\widehat q\widehat\mu=0$ can be explicitely written as
\begin{equation}\label{algformsode}
\partial_\tau^2\widehat \mu+\left(\frac12\partial_\tau(\ln\alpha)+\frac{\widehat p}{\sqrt{\alpha}}\right)\partial_\tau\widehat\mu+\left(\frac{\widehat q}{\alpha}\right)\widehat\mu=0.
\end{equation}
In case that $\sqrt{\alpha}$, $\widehat p$ and $\widehat q$ are rational functions in $\tau$, the equation \eqref{algformsode} is the algebraic form of the first one, i.e., the equation $\partial_t^2\mu+p\partial_t\mu+q\mu=0$ has been algebrized through a Hamiltonian change of variable. This procedure is called \emph{Hamiltonian Algebrization}, which  is an \emph{isogaloisian transformation}, i.e., the  differential Galois group is preserved under Hamiltonian Algebrization procedure. Further details and proofs can be found in \cite{ac,acmowe}.

\subsection{Propagators and Green Functions}

In this section, as well in the rest of the paper, we follow \cite{cosu}
considering the one-dimensional time-dependent Schr\"{o}dinger equation for
an harmonic oscillator
\begin{equation}
i\partial _{t}\psi =H\psi ,\quad H=a(t)p^{2}+b(t)x^{2}+c(t)(px+xp),\quad
p=-i\partial _{x}.  \label{SE0}
\end{equation}
The Schr\"{o}dinger equation \eqref{SE0} can be written as
\begin{equation}
i\partial _{t}\psi =\left(-a\left( t\right) \partial _{x}^{2} +b\left(
t\right) x^{2} -ic\left( t\right) -2ic\left( t\right) 
x\partial _{x}\right)\psi .  \label{SE1}
\end{equation}
We start considering the Riccati equation
\begin{equation}
\partial_t\alpha+b\left( t\right) +2c\left( t\right) \alpha +4a\left(
t\right) \alpha ^{2}=0,  \label{Riccati equation}
\end{equation}%
where $a(t)$, $b(t)$ and $c(t)$ are elements of a differential field $K$,
with coefficient field $\mathbb{C}$.  By Proposition \ref{prop:almp} we can transform the Riccati equation (\ref{Riccati equation}) through the change of variable
\begin{equation}
\alpha \left( t\right) =\frac{1}{4a\left( t\right) }\frac{\partial_t\mu \left( t\right) }{\mu \left( t\right) }-\frac{c(t)}{4a(t)},
\end{equation}%
into the second order differential equation
\begin{equation}
\partial_t^2\mu -\tau (t)\partial_t\mu+4\sigma (t)\mu =0,.
\label{Caractheristic equation}
\end{equation}%
where \begin{equation*}
\tau \left( t\right) =\frac{\partial_ta}{a}-4c,\qquad \sigma \left( t\right)
=ab+\frac{c^{2}}{2}+\frac{c}{4}\left( \frac{\partial_ta}{a}-\frac{\partial_tc}{c}\right) 
\end{equation*}

Moreover, by Remark \ref{rem:almp}, the differential Galois group of the differential equation (\ref{Caractheristic equation}) is
virtually solvable if and only if the Riccati
equation (\ref{Riccati equation}) has an algebraic solution over the
differential field $K$. Furthermore, by Proposition \ref{prop:almp}, departing from the differential
equation (\ref{Caractheristic equation}) we can arrive at the Riccati equation (\ref{Riccati equation})
through changes of variables. The following lemmas show how can we construct
propagators based on explicit solutions in (\ref{Riccati equation}) and (\ref%
{Caractheristic equation}).

\begin{lemma}\label{lemma1}
\cite{colosusu, cosu, susu2} Given $a(t),b(t)$ and $c(t)$ piecewise
continuous, there exists an interval $I$ of time where the following
(Riccati-type) system
\begin{align}
& \partial_t\alpha+b\left( t\right) +2c\left( t\right) \alpha +4a\left(
t\right) \alpha ^{2}=0,  \label{schr7} \\
& \partial_t\beta +\left( c\left( t\right) +4a\left( t\right) \alpha
\left( t\right) \right) \beta =0,  \label{schr8} \\
& \partial_t\gamma +a\left( t\right) \beta ^{2}\left( t\right) =0,
\label{schr9}
\end{align}%
has as a fundamental solution in terms of solutions of the following
(characteristic) equation%
\begin{equation}
\partial_t^2\mu-\tau \left( t\right) \partial_t\mu+4\sigma \left(
t\right) \mu =0  \label{charac}
\end{equation}%
\begin{equation}
\tau \left( t\right) =\frac{\partial_ta}{a}-4c,\qquad \sigma \left( t\right)
=ab+\frac{c^{2}}{2}+\frac{c}{4}\left( \frac{\partial_ta}{a}-\frac{\partial_tc}{c}\right)  \label{charac2}
\end{equation}
given by$:$
\begin{eqnarray}
&&\alpha _{0}\left( t\right) =\frac{1}{4a\left( t\right) }\frac{\partial_t\mu
_{0}\left( t\right) }{\mu _{0}\left( t\right) }-\frac{c\left(
t\right) }{4a\left( t\right) },  \label{A0} \\
&&\beta _{0}\left( t\right) =-\frac{1}{\mu _{0}\left( t\right) },\qquad
\label{B0} \\
&&\gamma _{0}\left( t\right) =\frac{1}{2\mu _{1}\left( 0\right) }\frac{\mu
_{1}\left( t\right) }{\mu _{0}\left( t\right) }+\frac{c\left( 0\right) }{%
2a\left( 0\right) }  \label{C0}
\end{eqnarray}
provided that $\mu _{0}$ and $\mu _{1}$ are standard solutions of (\ref%
{charac})-(\ref{charac2}) with $\mu _{0}\left( 0\right) =0,$ $\partial_t\mu
_{0}\left( 0\right) =2a\left( 0\right) \neq 0,$ and $\mu
_{1}\left( 0\right) \neq 0,$ $\partial_t\mu _{1}\left( 0\right) =0,$ and $%
\delta _{0}\left( 0\right) =-\varepsilon _{0}\left( 0\right) =g\left(
0\right) /\left( 2a\left( 0\right) \right) $, $\kappa _{0}\left( 0\right)
=0. $ Further, the following asymptotics hold:{\small
\begin{eqnarray}
&&\alpha _{0}\left( t\right) =\frac{1}{4a\left( 0\right) t}-\frac{c\left(
0\right) }{4a\left( 0\right) }-\frac{\partial_ta(0)}{8a^{2}(0)}+O(t), \\
&&\beta _{0}\left( t\right) =-\frac{1}{2a\left( 0\right) t}+\frac{\partial_ta(0)}{4a^{2}(0)}+O(t),\qquad \\
&&\gamma _{0}\left( t\right) =\frac{1}{4a\left( 0\right) t}+\frac{c\left(
0\right) }{4a\left( 0\right) }-\frac{\partial_ta(0)}{8a^{2}(0)}+O(t)
\end{eqnarray}%
}as $t\rightarrow 0$, with\ $a(t),b(t)$ and $c(t)$ sufficiently smooth.
\end{lemma}

\begin{lemma}\label{lemma2}
\cite{colosusu, cosu, susu2, su1} The Green function, or Feynman's
propagator, corresponding to the Schr\"{o}dinger equation \ref{SE0} can be
obtained as
\begin{equation}
\psi =G\left( x,y,t\right) =\frac{1}{\sqrt{2\pi i\mu _{0}\left( t\right) }}\
e^{i\left( \alpha _{0}\left( t\right) x^{2}+\beta _{0}\left( t\right)
xy+\gamma _{0}\left( t\right) y^{2}\right) },  \label{newgreen}
\end{equation}%
where $\alpha _{0}\left( t\right) ,$ $\beta _{0}\left( t\right) $ and $%
\gamma _{0}\left( t\right) $\ are solutions of the Riccati-type system. Then
the superposition principle allows us to solve the corresponding Cauchy
initial value problem:%
\begin{equation*}
\psi (x,t)=\int_{-\infty }^{\infty }G(x,y,t)\psi (y,0)dy
\end{equation*}%
for suitable data $\psi (x,0)$\bigskip $=\varphi (x).$ Further as $%
t\rightarrow 0,$%
\begin{eqnarray*}
G\left( x,y,t\right)  &\sim &\frac{1}{\sqrt{2\pi ia(0)t}}\exp \left[ i\frac{%
\left( x-y\right) ^{2}}{4a(0)t}\right]  \\
&&\times \exp \left[ -i\left( \frac{\partial_ta(0)}{8a^{2}(0)}\left(
x-y\right) ^{2}+\frac{c\left( 0\right) }{4a\left( 0\right) }\left(
x-y\right) ^{2}\right) \right].
\end{eqnarray*}%
\bigskip
\end{lemma}
\section{Galoisian Approach to Propagators}\label{Galois_prop}

In this section we apply the Picard-Vessiot theory in the context of
propagators.

We are interested in \emph{Liouvillian solutions} of the linear Schr\"{o}%
dinger equation (\ref{SE0}); that is, the so-called \emph{Liouvillian
propagator} (\ref{newgreen}) is obtained through Liouvillian functions. In
this way, we can give a Galoisian formulation for this kind of integrability.

\begin{definition}
\label{liouvillian_LSE} Linear Schr\"{o}dinger equation (\ref{SE0}) is
integrable in Galoisian sense (Galois integrable), when it has a Liouvillian
propagator (\ref{newgreen}).
\end{definition}

\begin{theorem}[Galoisian approach to LSE]
\label{Galois_LSE} Linear Schr\"{o}dinger equation (\ref{SE0}) is Galois
integrable if and only if the differential Galois group of the
characteristic equation (\ref{Caractheristic equation}) is virtually
solvable.
\end{theorem}

\begin{proof}
Consider $K$ as the differential field of the characteristic equation and
also consider $\mu _{0}$ and $\mu _{1}$ as in Lemma \ref{lemma1}. Let us
suppose that $\mu _{0}$ is a Liouvillian solution of the characteristic
equation; then by D'Alambert reduction method, the second solution $\mu _{1}$
is Liouvillian too. Thus, in virtue of Theorem \ref{kolchin}, we see that
the differential Galois group of the characteristic equation is virtually
solvable and $\alpha $ is an algebraic solution of the Riccati equation over 
$K$. Now, using $\mu _{0}$ and $\mu _{1}$ in $\beta $ and $\gamma $ such as
in Lemma \ref{lemma1}, we see that they are Liouvillian functions over $K$
and therefore the propagator is Liouvillian. Thus, by Definition \ref%
{liouvillian_LSE} we get that the Schr\"{o}dinger equation is Galois
integrable. Conversely, assuming that the Schr\"{o}dinger equation is Galois
integrable we see that it has a Liouvillian propagator that can be obtained
through an algebraic solution of the Riccati equation over $K$. For
instance, there exists a Liouvillian solution $\mu _{0}$ of the
characteristic equation which by Theorem \ref{kolchin} implies that its
differential Galois group is virtually solvable.
\end{proof}

\begin{remark}
In virtue of Theorem \ref{Galois_LSE} we can construct many Liouvillian
propagators through integrable second order differential equations over a
differential field as well by algebraic solutions, over such differential
field, of Riccati equations. This is the practical aim of this paper, which
will be given in Section \ref{newpropagators} and in Section \ref{toyexamples}.
\end{remark}

\section{Degenerate Parametric Oscillator and Ince's Equation by a Galoisian
Approach}\label{newpropagators}

\noindent In this section we use the Galoisian approach to LSE, Section \ref{Galois_prop}, Theorem \ref%
{Galois_LSE} to study the integrability of  Ince's equation 
\begin{equation}
\partial _{t}^{2}\mu +{\frac{2\lambda \,\omega \,\sin \left( 2\,\omega
\,t\right) }{\omega +\lambda \,\cos \left( 2\,\omega \,t\right) }}\partial
_{t}\mu +{\frac{{\omega }^{3}-3\,\omega \,{\lambda }^{2}-\left( {\omega }%
^{2}\lambda +{\lambda }^{3}\right) \cos \left( 2\,\omega \,t\right) }{\omega
+\lambda \,\cos \left( 2\,\omega \,t\right) }}\mu =0  \label{inceeq}
\end{equation}%
that corresponds to the characteristic equation (see (\ref{charac})-(\ref%
{charac2}) and Lemma  6 of section 2) of the parametric oscillator given by
the Schrodinger equation (\ref{SEsu1})-(\ref{SEsu2}).  To find the solutions of Ince's equation\eqref{inceeq} we apply two practical tools of differential Galois
theory: \emph{Kovacic Algorithm} developed by J. Kovacic in \cite{kov} and 
\emph{Hamiltonian Algebrization} developed by the first author in \cite%
{ac,acmowe}. Here we follow the version of Kovacic Algorithm given in \cite%
{ac,acmowe,kov}.

In this section we find solutions to

\subsection{Case 1 $\lambda =\omega$ }

Ince's equation \eqref{inceeq} becomes to 
\begin{equation}
\partial _{t}^{2}\mu +2\lambda \tan (\lambda t)\partial _{t}\mu -2\mu =0.
\label{cosutan}
\end{equation}%
To solve the characteristic equation \eqref{cosutan} through differential
Galois theory we apply the hamiltonian algebrization process and Kovacic's
algorithm to obtain the solutions as well its differential Galois group.
Afterward, we construct the corresponding propagator associated to this
characteristic equation. We consider $K=\mathbb{C}(\tan \lambda t)$ as the
differential field of the characteristic equation.

We first consider the hamiltonian change of variable $\tau =\tan \lambda t$
obtaining $\alpha =\lambda ^{2}(1+\tau ^{2})^{2}$ and using the Hamiltonian
Algebrization process we get 
\begin{equation}
\partial _{\tau }^{2}\hat{\mu}(\tau )+\frac{4\tau }{1+\tau ^{2}}\partial
_{\tau }\hat{\mu}(\tau )-\frac{2}{(1+\tau ^{2})^{2}}\hat{\mu}(\tau )=0,
\label{algeb}
\end{equation}%
i.e., the algebraic form of the characteristic equation, being $\widehat{K}=%
\mathbb{C}(\tau )$ its differential field which is isomorphic to $K=\mathbb{C%
}(\tan \lambda t)$. In order to apply the Kovacic algorithm we should reduce
the equation \eqref{algeb} through the change of dependent variable 
\begin{equation*}
\hat{\mu}(\tau )=y\exp \left( -\frac{1}{2}\int_{0}^{\tau }\frac{4s}{1+s^{2}}%
ds\right) =\frac{y}{1+\tau ^{2}}.
\end{equation*}%
Thus we obtain the reduced form 
\begin{equation}
\partial _{\tau }^{2}y=ry,\quad r=\frac{2\tau ^{2}+4}{(1+\tau ^{2})^{2}}=%
\frac{2}{1+\tau ^{2}}+\frac{2}{(1+\tau ^{2})^{2}},  \label{redformod2}
\end{equation}%
where the Picard-Vessiot extension is the same as in \eqref{algeb}, denoted
by $\widehat{L}$, and for instance the differential Galois group will be the
same for both equations due to the fact that their coefficients are rational
functions in $\tau $, i.e., $\widehat{K}=\mathbb{C}(\tau )$ and their
solutions are linked by a rational function in $\tau $ which does not
changes the Picard-Vessiot extension. \newline

Applying Kovacic's algorithm we see that $\Gamma =\{i,-i,\infty \}$, $\circ
r_{i}=\circ r_{-i}=\circ r_{\infty }=2$, which implies that equation %
\eqref{redformod2} can fall in cases 1, 2, 3 or 4 of the algorithm. We begin
analyzing case one.\newline

\noindent \emph{Case 1}. According to step 1, we check conditions ($c_{2}$)
and $(\infty _{2})$ to obtain $b_{i}=b_{-i}=-\frac{1}{2}$ and $b_{\infty }=2$%
. In this way, $\alpha _{i}^{\pm }=\alpha _{-i}^{\pm }=\frac{1\pm i}{2}$, $%
\alpha _{\infty }^{\pm }=\frac{1\pm 3}{2}$, $\sqrt{r}_{i}=\sqrt{r}_{-i}=%
\sqrt{r}_{\infty }=0$. By step 2 we obtain $D=\{1\}$ through two different
options: $1=\alpha _{\infty }^{+}-\alpha _{i}^{+}-\alpha _{-i}^{-}$ and $%
1=\alpha _{\infty }^{+}-\alpha _{i}^{-}-\alpha _{-i}^{+}$; for instance we
have two possibilities for $\omega $: $\omega _{1}=\frac{1}{\tau -i}-\frac{1%
}{\tau +i}$ and $\omega _{2}=\frac{1}{\tau +i}-\frac{1}{\tau -i}$, that is, $%
\omega _{1}^{2}=\omega _{2}^{2}$. Following step 3 we see that there does
not exist a polynomial $P_{1}(\tau )=\tau +a_{0}$ corresponding to this part
of the algorithm. In this way, we should go to the case 2 of Kovacic
algorithm due to $\mathrm{DGal}(\widehat{L}/\widehat{K})$ not being a
subgroup of the triangular group of $\mathrm{SL}(2,\mathbb{C})$.\newline

\noindent \emph{Case 2}. According to step 1, we check conditions ($c_{2}$)
and $(\infty _{2})$ to obtain $b_{i}=b_{-i}=-\frac{1}{2}$ and $b_{\infty }=2$%
. In this way, $E_{i}=E_{-i}=\{2\}$, $E_{\infty }=\{-4,2,8\}$. By step 2 we
obtain $D=\{2\}$ only for $e_{\infty }=8$. In this way, we obtain $\theta =%
\frac{1}{\tau -i}+\frac{1}{\tau +i}$. Following step 3 we obtain the
polynomial $P_{2}(\tau )=\tau ^{2}-1$, which leads us to obtain $\phi =\frac{%
1}{\tau -i}+\frac{1}{\tau +i}+\frac{2\tau }{\tau ^{2}-1}$. In this way,
solving the algebraic equation 
\begin{equation*}
\omega ^{2}+\phi \omega +\frac{1}{2}(\partial _{\tau }\phi +\phi ^{2}-2r)=0,
\end{equation*}%
we obtain two solutions for $\omega $: 
\begin{equation*}
\omega _{-}=2\,{\frac{{\tau }^{2}-\tau +1}{{\tau }^{3}-{\tau }^{2}+\tau -1}}%
,\quad \omega _{+}=2\,{\frac{{\tau }^{2}+\tau +1}{{\tau }^{3}+{\tau }%
^{2}+\tau +1}}.
\end{equation*}%
In this way, due to $y=e^{\int \omega d\tau }$, we have the general solution
of \eqref{redformod2}: 
\begin{equation*}
y=C_{1}e^{-\arctan \tau }(\tau -1)\sqrt{1+\tau ^{2}}+C_{2}e^{\arctan \tau
}(\tau +1)\sqrt{1+\tau ^{2}},
\end{equation*}%
for instance $\mathrm{DGal}(\widehat{L}/\widehat{K})=\mathbb{D}_{\infty }$,
that is, the infinite dihedral group. Now, the general solution for equation %
\eqref{algeb} is given by 
\begin{equation*}
\widehat{\mu }(\tau )=C_{1}{\frac{e^{-\arctan \tau }(\tau -1)}{\sqrt{1+\tau
^{2}}}}+C_{2}{\frac{e^{\arctan \tau }(\tau +1)}{\sqrt{1+\tau ^{2}}}};
\end{equation*}%
for instance the differential Galois group for the algebrized characteristic
equation \eqref{algeb} is also the dihedral infinite group $\mathbb{D}%
_{\infty }$. Recalling that $\tau =\tan \lambda t$, we get the general
solution of the characteristic equation 
\begin{equation*}
\mu (t)=C_{1}e^{-\lambda t}(\sin \lambda t-\cos \lambda t)+C_{2}e^{\lambda
t}(\sin \lambda t+\cos \lambda t),
\end{equation*}%
which can also be written as 
\begin{equation*}
\mu (t)=(C_{1}+C_{2})(\sinh \lambda t\cos \lambda t+\cosh \lambda t\sin
\lambda t
\end{equation*}%
\begin{equation*}
+(C_{2}-C_{1})\sinh \lambda t\sin \lambda t+\cosh \lambda t\cos \lambda t,
\end{equation*}%
and its differential Galois group is also the dihedral infinite group, i.e.,$%
\mathrm{DGal}(L/K)=\mathbb{D}_{\infty }$. Now, we find $\mu _{0}(t)$ and $%
\mu _{1}(t)$ satisfying the conditions of Lemma \ref{lemma1} 
\begin{equation*}
\mu _{0}(t)=\sinh \lambda t\cos \lambda +\cosh \lambda t\sin \lambda t,\quad
C_{1}=C_{2}=\frac{1}{2},
\end{equation*}%
\begin{equation*}
\mu _{1}(t)=\sinh \lambda t\sin \lambda t+\cosh \lambda t\cos \lambda
t,\quad -C_{1}=C_{2}=\frac{1}{2},
\end{equation*}%
which for $\lambda =1$ corresponds to the solutions given in \cite{cosuin}.
Therefore, by Theorem \ref{Galois_LSE} the Schr\"{o}dinger equation (\ref%
{SEsu1})-(\ref{SEsu2}) is Galois integrable and its propagator by Lemma 7 is
given by

\begin{eqnarray}
G\left( x,y,t\right)  &=&\frac{1}{\sqrt{2\pi i(\cos \lambda t\sinh \lambda
t+\sin \lambda t\cosh \lambda t)}}  \label{propagator} \\
&&\times \exp \left[ \frac{(x^{2}-y^{2})\sin \lambda t\sinh \lambda
t+2xy-(x^{2}-y^{2})\cos \lambda t\cosh \lambda t}{2i(\cos \lambda t\sinh
\lambda t+\sin \lambda t\cosh \lambda t)}\right]   \notag
\end{eqnarray}%
and it is Liouvillian.

On the other hand, after
the Hamiltonian Algebrization process over the equation \eqref{cosutan}, we use the command \texttt{kovacicsols%
} over the equation \eqref{algeb} to obtain 
\begin{equation*}
\left[ \frac{\left( \frac{i+\tau }{-\tau +i}\right) ^{\frac{i}{2}}(1+\tau )}{%
\sqrt{1+\tau ^{2}}},\frac{\left( \frac{-\tau +i}{i+\tau }\right) ^{\frac{i}{2%
}}(-1+\tau )}{\sqrt{1+\tau ^{2}}}\right] ,
\end{equation*}%
therefore we can write the general solution as 
\begin{equation}
\hat{\mu}=C_{1}\frac{\left( \frac{i+\tau }{-\tau +i}\right) ^{\frac{i}{2}%
}(1+\tau )}{\sqrt{1+\tau ^{2}}}+C_{2}\frac{\left( \frac{-\tau +i}{i+\tau }%
\right) ^{\frac{i}{2}}(-1+\tau )}{\sqrt{1+\tau ^{2}}}.  \label{solKov}
\end{equation}

Recalling that 
\begin{eqnarray*}
\cosh z+\sinh z &=&e^{z},\quad \cos \left( \arctan z\right) =-\frac{1}{\sqrt{%
1+z^{2}}},\quad  \\
\sin \left( \arctan z\right)  &=&\frac{z}{\sqrt{1+z^{2}}},\quad \arctan z=%
\frac{1}{2}i(\ln (1-iz)-\ln (1+iz)),
\end{eqnarray*}

for instance equation (\ref{solKov}) becomes 
\begin{eqnarray*}
\hat{\mu}(\tau ) &=&\left( C_{1}e^{\arctan \tau }-C_{2}e^{-\arctan \tau
}\right) \cos (\arctan \tau ) \\
&&+\left( C_{1}e^{\arctan \tau }+C_{2}e^{-\arctan \tau }\right) \sin
(\arctan \tau ),
\end{eqnarray*}

Now, recalling $\tau =\tan \lambda t,$ we obtain%
\begin{equation*}
{\mu }(t)=\left( C_{1}e^{\lambda t}-C_{2}e^{-\lambda t}\right) \cos \lambda
t+\left( C_{1}e^{\lambda t}+C_{2}e^{-\lambda t}\right) \sin \lambda t.
\end{equation*}

Finally, from the conditions of $\mu _{0}$ and $\mu _{1}$ stated in Lemma %
\ref{lemma1} ${\mu }_{0}(t)$ and $\mu _{1}(t)$ become 
\begin{equation*}
{\mu }_{0}(t)=\sinh \lambda t\cos \lambda t+\cosh \lambda t\sin \lambda t,
\end{equation*}%
\begin{equation*}
{\mu }_{1}(t)=\cosh \lambda t\cos \lambda t+\sinh \lambda t\sin \lambda t,
\end{equation*}%
that are the same solutions found using Kovacic Algorithm by hand.

\subsection{Case 2 $\lambda \neq \omega $ }

For the Ince's equation \eqref{inceeq} using Hamiltonian Algebrization
procedure and Kovacic Algorithm. By properties of double angle, we can write
the equation \eqref{inceeq} in terms of $\tan (\omega t)$, for instance, we
can consider as its differential field to $K=\mathbb{C}(\tan \omega t)$.
After the Hamiltonian change of variable $\tau =\tan \omega t$ we obtain $%
\alpha =\omega ^{2}(1+\tau ^{2})^{2}$ and by Hamiltonian Algebrization
procedure we get as algebraic form of \eqref{inceeq} to 
\begin{equation*}
\begin{array}{l}
\partial _{\tau }^{2}\widehat{\mu }+\varphi _{1}(\tau )\partial _{\tau }%
\widehat{\mu }+\varphi _{0}(\tau )\widehat{\mu }=0,\quad \varphi _{1}(\tau )=%
{\frac{2(\lambda -\omega )\tau ^{3}-(3\lambda +\omega )\tau }{(1+\tau
^{2})\left( (\lambda -\omega )\tau ^{2}-\lambda -\omega \right) }}, \\ 
\\ 
\varphi _{0}(\tau )=-{\frac{\left( {\omega }^{3}-3\,\omega \,{\lambda }^{2}+{%
\omega }^{2}\lambda +{\lambda }^{3}\right) {\tau }^{2}+{\omega }%
^{3}-3\,\omega \,{\ \lambda }^{2}-{\omega }^{2}\lambda -{\lambda }^{3}}{%
\left( 1+{\tau }^{2}\right) ^{2}\left( \left( \lambda -\omega \right) {\tau }%
^{2}-\lambda -\omega \right) {\omega }^{2}}}.%
\end{array}%
\end{equation*}%
We can eliminate one parameter through the change $\lambda =\kappa \omega $,
thus, our algebraic form becomes in 
\begin{equation}
\begin{array}{l}
\partial _{\tau }^{2}\widehat{\mu }+\varphi _{1}(\tau )\partial _{\tau }%
\widehat{\mu }+\varphi _{0}(\tau )\widehat{\mu }=0,\quad \varphi _{1}(\tau )=%
{\frac{2(\kappa -1)\tau ^{3}-(3\kappa +1)\tau }{(1+\tau ^{2})\left( (\kappa
-1)\tau ^{2}-\kappa -1\right) }}, \\ 
\\ 
\varphi _{0}(\tau )=-{\frac{(1-3\kappa ^{2}+\kappa +\kappa ^{3})\tau
^{2}+1-3\kappa ^{2}-\kappa -\kappa ^{3}}{(1+\tau ^{2})^{2}\left( (\kappa
-1)\tau ^{2}-\kappa -1\right) }},\quad \kappa \neq 1.%
\end{array}
\label{algformincgen}
\end{equation}%
Following the same steps for $\lambda =\omega $ in Kovacic Algorithm, by
Proposition \ref{prop:almp} we transform the equation \eqref{algformincgen}
in

\begin{enumerate}
\item 
\begin{equation}
\begin{array}{l}
\partial _{\tau }^{2}y=ry,\quad \widehat{\mu }(\tau )=y{\frac{\sqrt{(\kappa
-1){\tau }^{2}-1-\kappa }}{{1+{\tau }^{2}}}} \\ 
\\ 
r={\frac{\left( \left( -4{\kappa }^{3}-4\kappa +7\kappa ^{2}+{\kappa }%
^{4}\right) {\tau }^{4}+\left( 10\kappa ^{2}-2{\kappa }^{4}\right) {\ \tau }%
^{2}+4\kappa +7\kappa ^{2}+4{\kappa }^{3}+{\kappa }^{4}\right) }{\left( 1+{%
\tau }^{2}\right) ^{2}\left( \left( -1+\kappa \right) {\tau }^{2}-1-\kappa
\right) ^{2}}}.%
\end{array}
\label{alformincred}
\end{equation}%
We see that $\Gamma =\left\{ i,-i,\sqrt{\frac{\kappa +1}{\kappa -1}},-\sqrt{%
\frac{\kappa +1}{\kappa -1}},\infty \right\} $, $\circ r_{c}=2,\forall c\in
\Gamma $, which implies that equation \eqref{alformincred} could fall in
cases 1, 2, 3 or 4 of the algorithm. We discard the case one for similar
reasons when $\lambda =\omega $. 

In the same way as in the case $\lambda =\omega $, by step two and step
three we obtain the general solution of \eqref{alformincred}: 
\begin{equation}
y=C_{1}{e^{-\kappa \arctan \tau }(\tau -1)\sqrt{1+\tau ^{2}}\over \sqrt{(\kappa-1)\tau^2-\kappa-1}}+C_{2}{e^{\kappa
\arctan \tau }(\tau +1)\sqrt{1+\tau ^{2}}\over \sqrt{(\kappa-1)\tau^2-\kappa-1}},  \label{solgenince}
\end{equation}%
for instance $\mathrm{DGal}(\widehat{L}/\widehat{K})=\mathbb{D}_{\infty }$,
that is, the infinite dihedral group for any $\kappa \neq 0$. Now, the
general solution for equation \eqref{algformincgen} is given by 
\begin{equation}
\widehat{\mu }(\tau )=C_{1}{\frac{e^{-\kappa \arctan \tau }(\tau -1)}{\sqrt{%
1+\tau ^{2}}}}+C_{2}{\frac{e^{\kappa \arctan \tau }(\tau +1)}{\sqrt{1+\tau
^{2}}}};  \label{solgenince2}
\end{equation}%
for instance the differential Galois group for the algebrized characteristic
equation \eqref{algformincgen} is also the dihedral infinite group $\mathbb{D%
}_{\infty }$ for any value of $\kappa \neq 0$. Recalling that $\tau =\tan
\lambda t$ and $\lambda =\kappa \omega $, we get the general solution of the
characteristic equation 
\begin{equation*}
\mu (t)=C_{1}e^{-\lambda t}(\sin \omega t-\cos \omega t)+C_{2}e^{\lambda
t}(\sin \omega t+\cos \omega t),
\end{equation*}%
which can also be written as 
\begin{equation*}
\mu (t)=(C_{1}+C_{2})(\sinh \lambda t\cos \omega t+\cosh \lambda t\sin
\omega t
\end{equation*}%
\begin{equation*}
+(C_{2}-C_{1})\sinh \lambda t\sin \omega t+\cosh \lambda t\cos \omega t,
\end{equation*}%
and its differential Galois group is also the dihedral infinite group, i.e.,$%
\mathrm{DGal}(L/K)=\mathbb{D}_{\infty }$. Now, we find $\mu _{0}(t)$ and $%
\mu _{1}(t)$ satisfying the conditions of Lemma \ref{lemma1} 
\begin{equation*}
\mu _{0}(t)=\sinh \lambda t\cos \omega t+\cosh \lambda t\sin \omega t,\quad
C_{1}=C_{2}=\frac{1}{2},
\end{equation*}%
\begin{equation*}
\mu _{1}(t)=\sinh \lambda t\sin \omega t+\cosh \lambda t\cos \omega t,\quad
-C_{1}=C_{2}=\frac{1}{2},
\end{equation*}%
which for $\lambda =1$ corresponds to the solutions given in \cite{cosuin}.
Therefore, by Theorem \ref{Galois_LSE} the Schr\"{o}dinger equation (\ref%
{SEsu1})-(\ref{SEsu2}) is Galois integrable and its propagator by Lemma 7 is
given by

\begin{eqnarray}
G\left( x,y,t\right)  &=&\frac{1}{\sqrt{2\pi i(\cos \omega t\sinh \lambda
t+\sin \omega t\cosh \lambda t)}}  \label{propagator} \\
&&\times \exp \left[ \frac{(x^{2}-y^{2})\sin \omega t\sinh \lambda
t+2xy-(x^{2}-y^{2})\cos \omega t\cosh \lambda t}{2i(\cos \omega t\sinh
\lambda t+\sin \omega t\cosh \lambda t)}\right]   \notag
\end{eqnarray}%
and it's Liouvillian.

Now we can see that applying \texttt{kovacicsols} we obtain 
\begin{equation*}
\left[ \frac{\left( \frac{i+\tau }{-\tau +i}\right) ^{\frac{i}{2}\kappa
}(1+\tau )}{\sqrt{1+\tau ^{2}}},\frac{\left( \frac{-\tau +i}{i+\tau }\right)
^{\frac{i}{2}\kappa }(-1+\tau )}{\sqrt{1+\tau ^{2}}}\right] ,
\end{equation*}%
therefore we can write the general solution as 
\begin{equation}
\hat{\mu}=C_{1}\frac{\left( \frac{i+\tau }{-\tau +i}\right) ^{\frac{i}{2}%
\kappa }(1+\tau )}{\sqrt{1+\tau ^{2}}}+C_{2}\frac{\left( \frac{-\tau +i}{%
i+\tau }\right) ^{\frac{i}{2}\kappa }(-1+\tau )}{\sqrt{1+\tau ^{2}}},
\label{solKovgen}
\end{equation}%
due to $\kappa \arctan \tau =\kappa \frac{1}{2}i(\ln (1-i\tau )-\ln (1+i\tau
))$, the equation \eqref{solKovgen} becomes to 
\begin{equation*}
\tilde{\mu}=C_{1}{\frac{(\tau +1)e^{\kappa \arctan \tau }}{\sqrt{1+\tau ^{2}}%
}}+C_{2}{\frac{(\tau -1)e^{-\kappa \arctan \tau }}{\sqrt{1+\tau ^{2}}}},
\end{equation*}%
due to $\tau =\tan \omega t$ and $\kappa \omega =\lambda $, we arrive to 
\begin{equation*}
\mu =C_{1}(\sin \omega t+\cos \omega t)e^{\lambda t}+C_{2}(\sin \omega
t-\cos \omega t)e^{-\lambda t}.
\end{equation*}%
Finally, we can get the solutions $\mu _{0}$ and $\mu _{1}$ in the most
general way than in \cite{cosuin}, it is enough take $C_{1}=C2=\frac{1}{2}$
for $\mu _{0}$ and $C_{1}=-C_{2}=\frac{1}{2}$ for $\mu _{1}$, thus we obtain
such solutions for \eqref{inceeq} 
\begin{equation*}
\mu _{0}=\sin \omega t\cosh \lambda t+\cos \omega t\sinh \lambda t
\end{equation*}%
\begin{equation*}
\mu _{1}=\sin \omega t\sinh \lambda t+cos\omega t\cosh \lambda t.
\end{equation*}
\end{enumerate}

Thus, we have proven the following result.

\begin{theorem}
The fundamental solution of Ince's characteristic equation \eqref{inceeq} is
given by 
\begin{equation*}
\mu =C_{1}(\sin \omega t+\cos \omega t)e^{\lambda t}+C_{2}(\sin \omega
t-\cos \omega t)e^{-\lambda t}.
\end{equation*}%
and the propagator for the Schr\"{o}dinger equation (\ref{SEsu1})-(\ref%
{SEsu2}) is Galois integrable and its propagator by Lemma 7 is given by 
\begin{eqnarray}
G\left( x,y,t\right)  &=&\frac{1}{\sqrt{2\pi i(\cos \omega t\sinh \lambda
t+\sin \omega t\cosh \lambda t)}}  \label{propagator} \\
&&\times \exp \left[ \frac{(x^{2}-y^{2})\sin \omega t\sinh \lambda
t+2xy-(x^{2}-y^{2})\cos \omega t\cosh \lambda t}{2i(\cos \omega t\sinh
\lambda t+\sin \omega t\cosh \lambda t)}\right]   \notag
\end{eqnarray}%
and it is Liouvillian.
\end{theorem}

\section{Toy Examples}\label{toyexamples}
In this section we illustrate our Galoisian approach through some elementary examples. 
We use the Galoisian approach to LSE, Theorem \ref%
{Galois_LSE}, to generalize examples introduced in \cite{LaSu} as well as to
introduce toy examples. The starting point is the knowledge of the
integrability, in the Picard-Vessiot sense, of the characteristic equation,
or equivalently the existence of an algebraic solution over its differential
field of its associated Riccati equation. In general, to obtain the solutions of the
characteristic equations we can Kovacic Algorithm and Hamiltonian Algebrization procedure as in Section \ref{newpropagators}.

\subsubsection{Toy models inspired by integrable Riccati equations}

\begin{enumerate}
\item From the Riccati equation 
\begin{equation*}
\frac{d\alpha }{dt}+(\cos t)\alpha ^{2}=0
\end{equation*}
we can construct the propagator for the Schr\"{o}dinger equation 
\begin{equation*}
i\frac{\partial \psi }{\partial t}=-\frac{\cos t}{4}\frac{\partial ^{2}\psi 
}{\partial x^{2}}-id(t)\psi .
\end{equation*}
It has the characteristic equation $\mu (t)=2\sin t$ and the propagator is
given by (\ref{newgreen}) with 
\begin{equation*}
\alpha (t)=\frac{1}{\sin t},\qquad \beta (t)=-\frac{1}{2\sin t},\qquad
\gamma (t)=\frac{1}{16\sin t}.
\end{equation*}

\item From the Riccati equation 
\begin{equation*}
\frac{d\alpha }{dt}+2(t+a_{0})\alpha ^{2}=0
\end{equation*}
we can construct the propagator for the Schr\"{o}dinger equation 
\begin{equation*}
i\frac{\partial \psi }{\partial t}=-\left( \frac{t+a_{0}}{2}\right) \frac{%
\partial ^{2}\psi }{\partial x^{2}}-id(t)\psi ,
\end{equation*}
which has the characteristic equation $\mu (t)=t^{2}+2a_{0}t$ and the green
function given by (\ref{newgreen}) with 
\begin{equation*}
\alpha (t)=\frac{1}{t^{2}+2a_{0}t},\qquad \beta (t)=-\frac{1}{t^{2}+2a_{0}t}%
,\qquad \gamma (t)=\frac{1}{4\left( t^{2}+2a_{0}t\right) }.
\end{equation*}

\item From the Riccati equation 
\begin{equation*}
\frac{d\alpha }{dt}+2\cos t+\frac{1}{\cos t}\alpha ^{2}=0
\end{equation*}
we can construct the propagator for the Schr\"{o}dinger equation 
\begin{equation*}
i\frac{\partial \psi }{\partial t}=-\left( \frac{1}{4\cos t}\right) \frac{%
\partial ^{2}\psi }{\partial x^{2}}+2(\cos t)x^{2}\psi -id(t)\psi ,
\end{equation*}
which has the characteristic equation $\mu (t)=\frac{1}{2}\sin t$ and the
green function is given by (\ref{newgreen}) with 
\begin{equation*}
\alpha (t)=\frac{\cos ^{2}t}{\sin t},\qquad \beta (t)=-\frac{2}{\sin t},
\end{equation*}
\begin{eqnarray}
\gamma (t) &=&\frac{4}{\sin t\cos ^{2}t}-16\tan x\sec x+16\log \left( \cos 
\frac{x}{2}-\sin \frac{x}{2}\right) \\
&&-16\log \left( \sin \frac{x}{2}+\cos \frac{x}{2}\right) .
\end{eqnarray}

\item From the Riccati equation 
\begin{equation*}
\frac{d\alpha }{dt}+2(\tan t)\alpha +\frac{1}{4\cos t}\alpha ^{2}=0
\end{equation*}
we can construct the propagator for the Schr\"{o}dinger equation 
\begin{equation*}
i\frac{\partial \psi }{\partial t}=-\left( \frac{1}{4\cos t}\right) \frac{%
\partial ^{2}\psi }{\partial x^{2}}-i(\tan t)x\frac{\partial \psi }{\partial
x}-id(t)\psi ,
\end{equation*}
which has the characteristic equation $\mu (t)=\frac{1}{2}\sin t$ and the
green function given by (\ref{newgreen}) with 
\begin{equation*}
\alpha (t)=\frac{\cos ^{2}t}{\sin t},\qquad \beta (t)=-\frac{2}{\sin
t\left\vert \cos t\right\vert },\qquad \gamma (t)=\frac{1}{\sin t\cos ^{4}t}.
\end{equation*}

\item From the Riccati equation 
\begin{equation*}
\frac{d\alpha }{dt}-ae^{\lambda t}-ate^{\lambda t}\alpha -\alpha ^{2}=0
\end{equation*}
we deduce that for the Schr\"{o}dinger equation 
\begin{equation*}
i\frac{\partial \psi }{\partial t}=-\frac{1}{4}\frac{\partial ^{2}\psi }{%
\partial x^{2}}-ae^{\lambda t}x^{2}\psi +\frac{ate^{\lambda t}}{2}x\frac{%
\partial \psi }{\partial x}-id(t)\psi ,
\end{equation*}
the propagator is given by (\ref{newgreen}) where 
\begin{equation*}
\alpha (t)=-\frac{1}{t},
\end{equation*}

\begin{equation*}
\beta (t)=-\frac{1}{t}\exp \left( \frac{ae^{\lambda t}(\lambda t-1)}{%
2\lambda ^{2}}+\frac{a}{2\lambda ^{2}}\right) ,
\end{equation*}

\begin{equation*}
\gamma (t)=-\frac{1}{4t}\exp \left( -\frac{ae^{\lambda t}(\lambda t-1)}{%
\lambda ^{2}}-\frac{a}{\lambda ^{2}}\right) +\frac{a}{4\lambda }\left(
e^{\lambda t}-1\right) .
\end{equation*}
\end{enumerate}

\subsubsection{Characteristic equation $\partial_t^2\protect\mu +t^{n}%
\protect\mu =0$}

We consider in this equation as differential field to $K=\mathbb{C}(t)$.
This is a generalization of the case $n=1$, which was presented in \cite%
{LaSu}. However, the case $n=1$ does not correspond to Galoisian
integrability of the characteristic equation: the solutions are not
Liouvillian due to the fact that they are Airy functions. Now, we study the
integrability, in Galoisian sense, of this equation through Kovacic
Algorithm. As in \cite{stpr}, we obtain three conditions for $n$ to get
virtual solvability of the differential Galois group. It follows the
construction of the propagators related with these integrability conditions.

\begin{enumerate}
\item Let's consider the characteristic equation $\partial _{t}^{2}\mu +\mu
=0$ ($n=0$). A basis of solutions is given by $\mathcal{B}=\{\sin t,\cos t\}$%
. Thus, the Picard-Vessiot extension is given by $L=\mathbb{C}(t,e^{it})$
and for instance the differential Galois group is $\mathrm{DGal}(L/K)=%
\mathbb{G}_{m}$, that is, the diagonal group of $\mathrm{SL}(2,\mathbb{C})$.
By Proposition \ref{prop:almp}, through the change of variable $\alpha
=\partial _{t}\mu /\mu $, we obtain the Riccati equation 
\eqref{Riccati
equation}, where $a=1/4$, $b=1$ and $c=0$. Considering $\mu _{0}(t)=\lambda
_{1}\sin t$ and $\mu _{1}(t)=\lambda _{2}\cos t$, we see that the conditions
of Lemma \ref{lemma1} are satisfied when $\lambda _{1}=1/2$. Furthermore, $%
\lambda _{2}$ must be $2$ to get $W(\mu _{0},\mu _{1})=1$. In this way, by
Theorem \ref{Galois_LSE} the Schr\"{o}dinger equation is Galois integrable.

\item Let's consider the characteristic equation $\partial _{t}^{2}\mu +\mu
/t^{2}=0$ ($n=-2$). A basis of solutions is given by $\mathcal{B}%
=\{t^{m+1},t^{-m}\}$, being $m=(-1\pm \sqrt{5})/2$. Since $m\notin \mathbb{Q}
$, the Picard-Vessiot extension is $L=\mathbb{C}(t,t^{m})$ and therefore $%
\mathrm{DGal}(L/K)=\mathbb{G}_{m}$. In \cite{ac} there is a complete study
of the Galoisian structure of this equation. Thus, by Theorem \ref%
{Galois_LSE} the Schr\"{o}dinger equation is Galois integrable.

\item The characteristic equation $\partial _{t}^{2}\mu +\mu /t^{4}=0$ ($n=-4
$). A basis of solutions is given by $\mathcal{B}=\{t\cos (1/t),t\sin (1/t)\}
$. Thus, the Picard-Vessiot extension is $L=\mathbb{C}(t,e^{i/t})$ and
therefore $\mathrm{DGal}(L/K)=\mathbb{G}_{m}$. Here, initial conditions for
the characteristic equation are satisfied when $t\rightarrow 0^{+}$, that
is, $\mu _{0}(t)\rightarrow 0$ when $t\rightarrow 0^{+}$. For instace, by
Theorem \ref{Galois_LSE} the Schr\"{o}dinger equation is Galois integrable.
\end{enumerate}

\section*{Final Remarks}
This paper is an starting point to study the integrability of partial
differential equations in a more general sense through differential Galois
theory. With this approach we studied the linear Schr\"{o}dinger equation
corresponding to a generalized (quadratic) harmonic oscillator, where the main result is the obtaining of the general solution of the Ince's differential equation and for instance the Liouvillian propagator of a degenerate parametric oscillator, which generalizes the particular results obtained in \cite{cosuin,LaSu}. Although there are a plenty of papers
concerning to explicit solutions and harmonic oscillator (see \cite{GoVi}),
we recall that differential Galois theory can provide the Liouvillian
solutions of characteristic equations without previous knowledge of such
equations. This is different to construct the explicit propagators knowing
apriory the solutions of the Riccatti or characteristic equation, which can
open other possibilities to study propagator with special functions as
characteristic equations, for example, Heun equation.

\section*{Acknowledgments}

The first author is partially supported by the MICIIN/FEDER grant number
MTM2009--06973, by the Generalitat de Catalunya grant number 2009SGR859 and
by DIDI -- Universidad del Norte. The second author acknowledges being a
recipient of Becas Iberoamericanas, jovenes profesores e invetigadores
Santander Universidades during 2012. The second author greatly appreciates
the support by University of Puerto Rico and Universitat de Barcelona during
the academic visit to the latter in Spring 2012 where this project was born.
E. Suazo also was supported by the AMS-Simons Travel Grants, with support
provided by the Simons Foundation to finish this project and the
continuation of others. The authors thank Carles Sim\'o and Juan J.
Morales-Ruiz for their useful comments, suggestions and great hospitality.

\end{document}